\documentclass[11pt]{llncs}

\usepackage{fullpage}
\usepackage{epsfig,url}
\usepackage{xspace,subfigure}
\usepackage{algorithmic}
\usepackage{algorithm}

\newtheorem{thrm}{Theorem}

\newcommand{\tf}{\ensuremath{\mathit{tf}}}
\newcommand{\df}{\ensuremath{\mathit{df}}}
\newcommand{\idf}{\ensuremath{\mathit{idf}}}

\newcommand{\rc}{\ensuremath{\mathit{range\_count}}}
\newcommand{\rr}{\ensuremath{\mathit{range\_report}}}
\newcommand{\rqq}{\ensuremath{\mathit{range\_quantile}}}
\newcommand{\rnv}{\ensuremath{\mathit{range\_next\_value}}}
\newcommand{\rint}{\ensuremath{\mathit{range\_intersect}}}
\newcommand{\dlist}{\ensuremath{\mathit{doc\_listing}}}
\newcommand{\dfreq}{\ensuremath{\mathit{doc\_frequency}}}
\newcommand{\dint}{\ensuremath{\mathit{doc\_intersect}}}
\newcommand{\Tag}{\ensuremath{\mathit{Tag}}}
\newcommand{\expand}{\ensuremath{\mathit{expand}}}

\newcommand{\up}{\vspace*{-0.4cm}}
\newcommand{\OUTPUT}{\STATE \textbf{output~}}
\newcommand{\RET}{\STATE \textbf{return~}}
\newcommand{\dcc}[1]{}
\newcommand{\home}[1]{#1}

\newcommand{\Oh}[1]
    {\ensuremath{\mathcal{O} \hspace{-.5ex} \left( {#1} \right)}}

\def\+{\!+\!}
\def\-{\!-\!}

\def\pref(#1,#2){$#1$ is a prefix of $#2$}
\def\suff(#1,#2){$#1$ is a suffix of $#2$}
\def\reg(#1,#2){$#2$ is $#1$-regular}
\def\notreg(#1,#2){$#2$ is not $#1$-regular}

\def\eqref#1{(\ref{#1})}

\newlength{\onedigit}
\newlength{\onecomma}
\begin{document}

\title{New Algorithms on Wavelet Trees \\ and Applications to Information
	Retrieval
	\thanks{Early parts of this work appeared in SPIRE 2009 \cite{GPT09}
		and SPIRE 2010 \cite{NP10}.}}


\author{
Travis Gagie\inst{1}
\and
Gonzalo Navarro\inst{2}\thanks{Partially supported by Fondecyt Grant 1-080019,
			Chile.}
\and
Simon J. Puglisi\inst{3}\thanks{Partially supported by the Australian Research
Council.}
}

\institute{
    Department of Computer Science,
    Aalto University, Finland.
    \email{travis.gagie@gmail.com} \\[1ex]
\and
    Department of Computer Science,
    University of Chile, Chile.
    \email{gnavarro@dcc.uchile.cl} \\[1ex]
 \and
    School of Computer Science and Information Technology,
    Royal Melbourne Institute of Technology, Australia.
    \email{\{simon.puglisi\}@rmit.edu.au}
}


\maketitle \thispagestyle{empty}

\begin{abstract}
Wavelet trees are widely used in the representation of sequences, permutations,
text collections, binary relations, discrete points, and other succinct data
structures. We show, however, that this still falls short of exploiting all of
the virtues of this versatile data structure. In particular we show how to
use wavelet trees to solve fundamental algorithmic problems such as
{\em range quantile} queries, {\em range next value} queries, and {\em range
intersection} queries. We explore several applications of these queries in
Information Retrieval, in particular {\em document retrieval} in
hierarchical and temporal documents, and in the representation of {\em
inverted lists}.
\end{abstract}

\section{Introduction}

The {\em wavelet tree} \cite{GGV03} is a versatile data structure that stores a
sequence $S[1,n]$ of elements from a symbol universe $[1,\sigma]$ within
asymptotically the same space required by a plain representation of the
sequence, $n\log\sigma\,(1+o(1))$ bits.%
\footnote{Our logarithms are in base 2 unless otherwise stated. Moreover,
within a time complexity, $\log x$ should be understood as $\max(1,\log x)$.}
Within that space, the wavelet tree is able to return any sequence
element $S[i]$, and also to answer two queries on $S$ that are fundamental in
compressed data structures for text retrieval:
\begin{eqnarray*}
rank_c(S,i) &~=~& \textrm{number of occurrences of symbol }c\textrm{ in }S[1,i], \\
select_c(S,j) &~=~& \textrm{position of the }j\textrm{th occurrence of symbol
}c\textrm{ in }S.
\end{eqnarray*}

The time for these three queries is $\Oh{\log \sigma}$.\footnote{This can be
reduced to $\Oh{1+\frac{\log\sigma}{\log\log n}}$ \cite{fmmn2007} using
multiary wavelet trees, but these do not merge well with the new algorithms we
develop in this article.}
Originally designed for compressing suffix arrays \cite{GGV03}, the usefulness
of the wavelet tree for many other scenarios was quickly realized. It was soon
adopted as a fundamental component of a large class of compressed
text indexes, the FM-index family, giving birth to most of its modern variants
\cite{FMMN04,MN05,fmmn2007,MN07implicit}.

The connection between the wavelet tree and an old geometric structure by
Chazelle \cite{Cha88} made it evident that wavelet trees could be used for {\em range
counting and reporting} points in the plane. More formally, given a set of $t$ points
$P=\{(x_i,y_i),~1\le i\le t\}$ on a discrete grid $[1,n] \times [1,\sigma]$, wavelet trees
answer the following basic queries:
\begin{eqnarray*}
\rc(P,x^s,x^e,y^s,y^e) &~=~& \textrm{number of pairs }(x_i,y_i)\textrm{
such that }x^s\le x_i\le x^e,~y^s\le y_i\le y^e, \\
\rr(P,x^s,x^e,y^s,y^e) &~=~& \textrm{list of those pairs }(x_i,y_i)
\textrm{ in some order,}
\end{eqnarray*}
both in $\Oh{\log\sigma}$ time \cite{MN06}.%
\footnote{Again, this can be reduced to $\Oh{1+\frac{\log \sigma}{\log\log n}}$
using multiary wavelet trees \cite{BHMM09}.}
These new capabilities were subsequently used to design powerful succinct
representations of two-dimensional point grids \cite{MN06,BHMM09,BLNS10},
permutations \cite{BN09}, and binary relations \cite{BCN10}, with applications
to other compressed text indexes \cite{Nav03,chsv2008,CN09}, document retrieval
problems \cite{vm2007} and many others.

In this paper we show, by uncovering new capabilities, that the full potential of wavelet trees is far from
realized. We show that the wavelet tree allows us to
solve the following fundamental queries:
\begin{eqnarray*}
\rqq(S,i,j,k) &~=~& k\textrm{th smallest value in }S[i,j], \\
\rnv(S,i,j,x) &~=~& \textrm{smallest }S[r]\ge x
				\textrm{ such that } i\le r\le j, \\
\rint(S,i_1,j_1,\ldots,i_k,j_k) &~=~& \textrm{distinct common values in }S[i_1,j_1], S[i_2,j_2], \ldots, S[i_k,j_k].
\end{eqnarray*}

The first two are solved in time $\Oh{\log\sigma}$, whereas the cost of the
latter is $\Oh{\log\sigma}$ per delivered value plus the size of the
intersection of the tries that describe the different values in $S[i_1,j_1]$
and $S[i_2,j_2]$. A crude upper bound for the latter is $\Oh{\min(\sigma,
j_1-i_1+1,j_2-i_2+1)}$, however, we give an adaptive analysis of our method,
showing it requires $O(\alpha\log \frac{\sigma}{\alpha})$ time, where $\alpha$
is the so-called {\em alternation complexity} of the problem~\cite{BK02}.

All these algorithmic problems are well known. Har-Peled and
Muthukrishnan~\cite{HM08} describe applications of range median queries
(a particular case of $\rqq$) to the analysis of Web advertizing
logs. Stolinski et al.~\cite{SGB10} use them for noise reduction in grey scale
images. Similarly, Crochemore et al.~\cite{CIR07} use $\rnv$
queries for interval-restricted pattern matching, and Keller et
al.~\cite{KKL07} and Crochemore et al.~\cite{CIKRW08} use them for many other
sophisticated pattern matching problems. Hon et al.~\cite{HSTV10} use
$\rint$ queries for generalized document retrieval, and
in a simplified form the problem also appears when processing conjunctive
queries in inverted indexes.


We further illustrate the importance of these fundamental algorithmic problems
by uncovering new applications in several Information Retrieval (IR) activities.
We first consider {\em document retrieval} problems on general sequences. This
generalizes the classical IR problems usually dealt with on Natural Language
(NL), and defines them in a more general setting where one has a collection $C$
of strings (i.e., the documents), and queries are strings as well. Then one is
interested in any substring of the collection that matches the query, and the
following IR problems are defined (among several others):
\begin{eqnarray*}
\dlist(q) &~=~& \textrm{distinct documents where query }q\textrm{ appears}, \\
\dfreq(q,d)&~=~& \textrm{number of occurrences of query }q\textrm{ in document }d,\\
\dint(q_1,\ldots,q_k) &~=~& \textrm{distinct documents where all queries
}q_1,\ldots,q_k\textrm{ appear.}
\end{eqnarray*}

These generalized IR problems have applications in text databases where the
concept of {\em words} does not exist or is difficult to define, such as in
Oriental languages, DNA and protein sequences, program code, music and other
multimedia sequences, and numeric streams in general. The interest in carrying
out IR tasks on, say, Chinese or Korean is obvious despite the difficulty of
automatically delimiting the words. In those cases one resorts to a model
where the text is seen as a sequence of symbols and must be able to retrieve
any substring. Agglutinating languages such as Finnish or German present
similar problems to a certain degree. While indexes for plain string matching
are well known, supporting more sophisticated IR tasks such as ranked document
retrieval is a very recent research area. It is not hard to imagine that
similar capabilities would be of interest in other types of sequences: for
example listing the functions where two given variables are used
simultaneously in a large software development system, or ranking a set of
gene sequences by the number of times a given substring marker occurs.

By constructing a {\em suffix array} $A$ \cite{MM93} on the text collection, one
can obtain in time $\Oh{|q|\log|C|}$ the range of $A$ where all the occurrence
positions
of $q$ in $C$ are listed. The classical solution to document retrieval problems
\cite{Mut02} starts by defining a {\em document array} $D$ giving the document
to which each suffix of $A$ belongs. Then problems like document listing boil
down to listing the distinct values in a range of $D$, and intersection of
documents becomes the intersection of values in a range of $D$. Both are
solved with our new fundamental algorithms (the former with range quantile
queries). Other queries such as computing frequencies reduce to a pair of
$rank_d$ queries on $D$.

Second, we generalize document retrieval problems to other scenarios. The first scenario
is {\em temporal} documents, where the document numbers are consistent
with increasing version numbers of the document set. Then one is interested in
restricting the above queries to a given interval of time (i.e., of
document numbers). A similar case is that of {\em hierarchical} documents,
which contain each other as in the case of an XML collection or a file system.
Here, restricting the query to a range of document numbers is equivalent to
restricting it to a subtree of the hierarchy. However, one can consider more
complex queries in the hierarchical case, such as marking a set of
{\em retrievable} nodes at query time and carrying out the operations with
respect to those nodes. We show how to generalize our algorithms to handle
this case as well.

Finally, we show that variants of our new fundamental algorithms are useful to
enhance the functionality of {\em inverted lists}, the favorite data structures
for both {\em ranked} and {\em full-text retrieval} in NL. Each of these
retrieval paradigms requires a different variant of the inverted list, and one
has to maintain both in order to support all the activities usually required in
an IR system. We show
that a wavelet tree representation of the inverted lists supports not only the
basic functionality of both representations within essentially the space of
one, but also several enhanced functionalities such as on-the-fly stemming and
restriction of documents, and most list intersection algorithms.

The article is structured as follows. In Section~\ref{sec:wt} we review the
wavelet tree data structure and its basic algorithmics. Section~\ref{sec:ir}
reviews some basic IR concepts. Then Section~\ref{sec:alg} describes the new
solutions to fundamental algorithmic problems, whereas
Sections~\ref{sec:doclist} and \ref{sec:invlists} explore applications to
various IR problems. Finally we conclude in Section~\ref{sec:concl}.

\section{Wavelet Trees}
\label{sec:wt}

A {\em wavelet tree} $T$ \cite{GGV03} for a sequence $S[1,n]$ over an ordered
alphabet $[1,\sigma]$ is an ordered, strictly
binary tree whose leaves are labeled with the distinct symbols in $S$ in
order from left to right, and whose internal nodes $T_v$ store binary strings
$B_v$. The binary string at the root contains $n$ bits and each is set to 0 or
1 depending on whether the corresponding character of $S$ is the label of a
leaf in $T$'s left or right subtree. For each internal node $v$ of $T$, the
subtree $T_v$ rooted at $v$ is itself a wavelet tree for the {\em subsequence}
$S_v$ of $S$ consisting of the occurrences of its leaf labels in $T_v$.
For example, if \(S = \mathsf{abracadabra}\) and the
leaves in $T$'s left subtree are labeled {\sf a}, {\sf b} and {\sf c}, then
the root stores \(00100010010\), the left subtree is a wavelet tree for
{\sf abacaaba} and the right subtree is a wavelet tree for {\sf rdr}.

In this article we consider balanced wavelet trees, where the number of
leaves to the left and to the right of each node differ at most by 1.
The important properties of such a wavelet tree for our purposes are summarized
in the following lemma.

\begin{lemma} \label{lem:wt}
The wavelet tree $T$ for a sequence $S[1,n]$ on alphabet $[1,\sigma]$ with
$u$ distinct symbols requires at most $n\log\sigma + \Oh{n}$ bits of space,
and can be constructed in $\Oh{n\log u}$ time.
\end{lemma}
\begin{proof}
By the description above the wavelet tree has height $\lceil \log u \rceil$ and
can be easily built in time $\Oh{n\log u}$ (we need to determine the $u \le
\min(n,\sigma)$ distinct values first, but this is straightforward within the
same complexity).

As for the space, note that the wavelet tree stores only the
bitmaps $B_v$ for all the nodes. The total length of the binary strings is at
most $n$ at each level of the wavelet tree, which adds up to
$n \lceil \log u \rceil$. Apart from the bitmaps, there is the binary tree of
$\Oh{u}$ nodes. Instead of storing the nodes, one can concatenate all the
bitmaps of the same depth and simulate the nodes \cite{MN06}, so this requires
just one pointer per level, $\Oh{\log u \log n} = o(n)$ bits.

The distinct values must be stored as well.
Indeed, if $\sigma \le n$, we can just assume all the $\sigma$
values exist and the wavelet tree will have $\lceil \log\sigma\rceil$ levels
and the theorem holds. Otherwise, we can mark the unique values in a bitmap
$U[1,\sigma]$, which can be stored in compressed form \cite{OS07} so that it
requires $u\log\frac{\sigma}{u}+\Oh{u}$ bits and the $i$th distinct number is
retrieved as $select_1(U,i)$ in constant time\footnote{For this, one has to
use a constant-time data structure for $select$ \cite{Mun96}
in their internal bitmap $H[1,2u]$ \cite{OS07}.}. Adding up all
the spaces we get $n\log u + \Oh{n} + u\log\frac{\sigma}{u}+\Oh{u} \le
n\log\sigma + \Oh{n}$ bits, and the construction time is $\Oh{u}$.

Finally, we can represent the bitmaps with data structures that support
constant-time (binary) $rank$ and $select$ operations \cite{Pat08}. The overall
extra space stays within $\Oh{n}$ bits and the construction time within
$\Oh{n\log u}$. Binary $rank$ and $select$ operations are essential to operate
on the wavelet trees, as seen shortly.
\qed
\end{proof}

The most basic operation of $T$ is to replace $S$, by retrieving any $S[i]$
value in $\Oh{\log u}$ time. The algorithm is as follows. We first examine
the $i$th bit of the root bitmap $B_{root}$. If $B_{root}[i]=0$, then symbol
$S[i]$ corresponds to a leaf descending by the left child of the root, and by
the right otherwise. In the first case we continue recursively on the left
child, $T_l$. However, position $i$ must now be mapped to the subsequence
handled at $T_l$. Precisely, if the 0 at $B_{root}[i]$ is the $j$th 0 in
$B_{root}$, then $S[i]$ is mapped to $S_v[j]$. In other words, when we
go left, we must recompute $i \leftarrow rank_0(B_{root},i)$. Similarly,
when we go right we set $i \leftarrow rank_1(B_{root},i)$.

When the tree nodes are not explicit, we find out the intervals corresponding
to $B_v$ in the levelwise bitmaps as follows. $B_{root}$ is a single bitmap.
If node $v$ has depth $d$, and $B_v$ corresponds to interval $B_d[l,r]$, then
its left child corresponds to $B_{d+1}[l,k]$ and its right child to
$B_{d+1}[k+1,r]$, where $k = rank_0(B_d,r)-rank_0(B_d,l-1)$ \cite{MN06}.

The wavelet tree can also answer $rank_c(S,i)$ queries on $S$ with a mechanism
similar to that for retrieving $S[i]$. This time one decides whether to go
left or right depending on which subtree of the current node the leaf labeled
$c$ appears in, and not on the bit values of $B_v$. The final $i$ value when
one reaches the leaf is the answer. Again, the process requires $\Oh{\log u}$ time.

Finally, $select_c(S,j)$ is also supported in $\Oh{\log u}$ time using
the wavelet tree. This time we start from position $j$ at the leaf labeled
$c$;\footnote{If the tree nodes are not explicitly stored then we first
descend to the node labeled $c$ in order to delimit the interval corresponding
to the leaf and to all of its ancestors in the levelwise bitmaps.}
this indeed corresponds to the $j$th occurrence of symbol $c$ in $S$.
If the leaf is a left child of its parent $v$, then the position of that $c$
in $S_v$ is $select_0(B_v,j)$, and $select_1(B_v,j)$ if the leaf is a right
child of $v$. We continue recursively from this new $j$ value until reaching
the root, where $j$ is the answer.

\begin{algorithm}[t]
\caption{Basic wavelet tree algorithms: On the wavelet tree of sequence $S$,
$\mathbf{access}(v_{root},i)$ returns $S[i]$; $\mathbf{rank}(v_{root},c,i)$
returns $rank_c(S,i)$; and $\mathbf{select}(v_{root},c,i)$ returns
$select_c(S,i)$.}
\label{alg:accrnksel}
\begin{tabular}{ccc}
\begin{minipage}{0.33\textwidth}
$\mathbf{access}(v,i)$
\begin{algorithmic}
\IF{$v$ is a leaf}
   \RET $label(v)$
\ELSIF{$B_v[i]=0$}
   \RET $\mathbf{access}(v_l,rank_0(B_v,i))$
\ELSE
   \RET $\mathbf{access}(v_r,rank_1(B_v,i))$
\ENDIF
\home{\STATE}
\home{\STATE}
\end{algorithmic}
\end{minipage}
&
\begin{minipage}{0.33\textwidth}
$\mathbf{rank}(v,c,i)$
\begin{algorithmic}
\IF{$v$ is a leaf}
   \RET $i$
\ELSIF{$c \in labels(v_l)$}
   \RET $\mathbf{rank}(v_l,c,rank_0(B_v,i))$
\ELSE
   \RET $\mathbf{rank}(v_r,c,rank_1(B_v,i))$
\ENDIF
\home{\STATE}
\home{\STATE}
\end{algorithmic}
\end{minipage}
&
\begin{minipage}{0.33\textwidth}
$\mathbf{select}(v,c,i)$
\begin{algorithmic}
\IF{$v$ is a leaf}
   \RET $i$
\ELSIF{$c \in labels(v_l)$}
   \RET $select_0(B_v,\mathbf{select}(v_l,c,i))$
\ELSE
   \RET $select_1(B_v,\mathbf{select}(v_r,c,i))$
\ENDIF
\end{algorithmic}
\end{minipage}
\end{tabular}
\end{algorithm}

Algorithm~\ref{alg:accrnksel} gives pseudocode for the basic $access$, $rank$ and
$select$ algorithms on wavelet trees. For all the pseudocodes in this article
we use the following notation: $v$ is a wavelet tree node and $v_{root}$ is the
root node. If $v$ is a leaf then its symbol is $labels(v) \in [1,\sigma]$.
Otherwise $v_l$ and $v_r$ are its left and right children, respectively, and
$B_v$ is its bitmap. For all nodes, $labels(v)$ is the range of leaf labels
that descend from $v$ (a singleton in case of leaves).

\begin{algorithm}[t]
\caption{Range algorithms: $\mathbf{count}(v_{root},x^s,x^e,[y^s,y^e])$ returns
$\rc(P,x^s,x^e,y^s,y^e)$ on the wavelet tree of sequence $P$; and
$\mathbf{report}(v_{root},x^s,x^e,[y^s,y^e])$ outputs all pairs $(y,f)$, where
$y^s \le y \le y^e$ and $y$ appears $f>0$ times in $P[x^s,y^s]$.}
\label{alg:range}
\begin{tabular}{cc}
\begin{minipage}{0.5\textwidth}
$\mathbf{count}(v,x^s,x^e,rng)$
\begin{algorithmic}
\IF{$x^s > y^s ~\lor~ labels(v) \cap rng = \emptyset$}
	\RET 0
\ELSIF{$label(v) \subseteq rng$}
	\RET $x^e-x^s+1$
\ELSE
	\STATE $x^s_l \leftarrow rank_0(B_v,x^s-1)+1$
	\STATE $x^e_l \leftarrow rank_0(B_v,x^e)$
	\STATE $x^s_r \leftarrow x^s-x^s_l$, $x^e_r \leftarrow x^e-x^e_l$
	\RET $\mathbf{count}(v_l,x^s_l,x^e_l,rng) +$
	\STATE \hspace{1.1cm} $\mathbf{count}(v_r,x^s_r,x^e_r,rng)$
\ENDIF
\end{algorithmic}
\end{minipage}
&
\begin{minipage}{0.5\textwidth}
$\mathbf{report}(v,x^s,x^e,rng)$
\begin{algorithmic}
\IF{$x^s > y^s ~\lor~ labels(v) \cap rng = \emptyset$}
	\RET
\ELSIF{$v$ is a leaf}
	\OUTPUT $(label(v),x^e-x^s+1)$
\ELSE
	\STATE $x^s_l \leftarrow rank_0(B_v,x^s-1)+1$
	\STATE $x^e_l \leftarrow rank_0(B_v,x^e)$
	\STATE $x^s_r \leftarrow x^s-x^s_l$, $x^e_r \leftarrow x^e-x^e_l$
	\STATE $\mathbf{report}(v_l,x^s_l,x^e_l,rng)$
	\STATE $\mathbf{report}(v_r,x^s_r,x^e_r,rng)$
\ENDIF
\end{algorithmic}
\end{minipage}
\end{tabular}
\end{algorithm}

As we make use of $\rc$ and a form of $\rr$ queries in this
article, we give pseudocode for them as well, in Algorithm~\ref{alg:range}.
Indeed, $\rc$ is a kind of multi-symbol $rank$ and $\rr$ is
a kind of multi-symbol $access$.

In Section~\ref{sec:alg} we develop new algorithms based on wavelet trees to
solve fundamental algorithmic problems. We prove now a few simple lemmas
that are useful for analyzing $\rc$ and $\rr$, as well as
many other algorithms we introduce throughout the article. Most results are
folklore but we reprove them here for completeness.

\begin{lemma} \label{lem:cover}
Any contiguous range of $\ell$ leaves in a wavelet tree is the set of
descendants of $\Oh{\log \ell}$ nodes.
\end{lemma}
\begin{proof}
Start with the $\ell$ leaves. For each consecutive pair that shares the same
parent, replace the pair by their parent. At most two leaves are not replaced,
and at most $\ell/2$ parents are created. Repeat the operation at the parent
level, and so on. After working on $\lceil\log\ell\rceil$ levels, we have at
most two nodes per wavelet tree level, for a total of $\Oh{\log\ell}$ nodes
covering the original interval.
\end{proof}

\begin{lemma} \label{lem:nodeset}
Any set of $r$ nodes in a wavelet tree of $u$ leaves has at most
$\Oh{r \log\frac{u}{r}}$ ancestors.
\end{lemma}
\begin{proof}
Consider the paths from the root to each of the $r$ nodes. They cannot be
all disjoint. They share the least if they diverge from depth
$\lceil\log r\rceil$. In this case, all the $\Oh{r}$ tree nodes of depth
up to $\lceil\log r\rceil$ belong to some path, and from that depth each of
the $r$ paths is disjoint, adding at most
$\lceil\log u\rceil-\lceil\log r\rceil$ distinct ancestors. The total is
$\Oh{r + r\log\frac{u}{r}}$.
\end{proof}

\begin{lemma} \label{lem:noderange}
Any set of $r$ nodes covering a contiguous range of leaves in a wavelet tree
of $u$ leaves has at most $\Oh{r + \log u}$ ancestors.
\end{lemma}
\begin{proof}
We first count all the ancestors of the $\ell$ consecutive leaves covered and
then subtract
the sizes of the subtrees rooted at the $r$ nodes $v_1, v_2, \ldots, v_r$.
Start with $\ell$ leaves. Mark all the parents of the leaves. At most
$\lceil \ell/2\rceil < 1 + \ell/2$ distinct parents are marked, as most pairs
of consecutive leaves will share the same parent. Mark the parents of the
parents. At most $\lceil (1+\ell/2)/2\rceil < 3/2 + \ell/4$ parents of
parents are marked.
At height $h$, the number of marked nodes is always less than $2 + \ell/2^h$.
Adding over all heights, we have that the total number of ancestors is at
most $2\ell + 2\log u$. Now let $\ell_i$ be the number of leaves covered by node
$v_i$, so that $\sum_{1\le i\le r} \ell_i = \ell$. The subtree rooted at each
$v_i$ has $2\ell_i-1$ nodes. By subtracting those subtree sizes and adding back
the $r$ root nodes we get $2\ell + 2\log u - (2\ell-r)+r = \Oh{r+\log u}$.
\end{proof}

From the lemmas we conclude that $\mathbf{count}$ in Algorithm~\ref{alg:range}
(left) takes time $\Oh{\log u}$: it finds the $\Oh{\log (y^e-y^s+1)}$ nodes
that cover the range $[y^s,y^e]$ (Lemma~\ref{lem:cover}), by working in time
proportional to the number of ancestors of those nodes, $\Oh{\log (y^e-y^s+1)
+\log u} = \Oh{\log u}$ (Lemma~\ref{lem:noderange}). Interestingly,
$\mathbf{report}$ in Algorithm~\ref{alg:range} (right) can be analyzed in two
ways. On one hand, it takes time $\Oh{y^e-y^s+\log u}$ as it arrives at most
at the $y^e-y^s+1$ consecutive leaves and thus it works on all of their
ancestors (Lemma~\ref{lem:noderange}). On the other hand, if it outputs $r$
results (which are not necessarily consecutive), it also works proportionally
to the number of their ancestors, $\Oh{r\log\frac{u}{r}}$
(Lemma~\ref{lem:nodeset}). The latter is an {\em output-sensitive} analysis.
The following lemma shows that the cost is indeed
$\Oh{\log u + r\log\frac{y^e-y^s+1}{r}}$.

\begin{lemma} \label{lem:leafrangeset}
The number of ancestors of $r$ wavelet tree leaves chosen from $\ell$
contiguous leaves, on a wavelet tree of $u$ leaves, is $\Oh{\log u +
r\log\frac{\ell}{r}}$.
\end{lemma}
\begin{proof}
By Lemma~\ref{lem:cover} those leaves are covered by $c=\Oh{\log \ell}$ nodes.
Say that $r_i$ of the $r$ searches fall within the $i$th of those subtrees,
then by Lemma~\ref{lem:nodeset} the number of nodes accessed within that
subtree is at most $\Oh{r_i\log\frac{\ell}{r_i}}$, adding up by convexity to at
most $\Oh{r\log\frac{\ell}{r/c}}$. Given the limit on $c$ this is
$\Oh{r\log\frac{\ell}{r}}$. The ancestors reached above those subtrees are
$\Oh{c + \log u} = \Oh{\log u}$ by Lemma~\ref{lem:noderange}, for a total of
$\Oh{\log u + r\log\frac{\ell}{r}}$.
\end{proof}

\section{Information Retrieval Concepts}
\label{sec:ir}

\subsection{Suffix and Document Arrays}
\label{sec:basicdoc}

Let $C$ be a collection of {\em documents} (which are actually strings over
an alphabet $[1,\sigma]$)
$D_1, D_2, \ldots, D_m$. Assume strings are terminated by a special
character ``$\$$'', which does not occur elsewhere in the collection. Now
we identify $C$ with the concatenation of all the documents,
$C[1,n]=D_1D_2\ldots D_m$. Each position $i$ defines a {\em suffix}
$C[i,n]$. A {\em suffix
array} \cite{MM93} of $C$ is an array $A[1,n]$ where the integers $[1,n]$ are
ordered in such a way that the suffix starting at $A[i]$ is lexicographically
smaller than that starting at $A[i+1]$, for all $1 \le i < n$.

Put another way, the suffix array lists all the suffixes of the collection in
lexicographic order. Since any substring of $C$ is the prefix of a suffix,
finding the occurrences of a query string $q$ in $C$ is equivalent to finding the
suffixes that start with $q$. These form a lexicographic range of suffixes,
and thus can be found via two binary searches in $A$ (accessing $C$ for the
string comparisons). As each
step in the binary search may require comparing up to $|q|$ symbols, the
total search time is $\Oh{|q|\log n}$. Once the interval $A[sp,ep]$ is
determined, all the occurrences of $q$ start at $A[i]$ for $sp \le i \le ep$.
Compressed full-text self-indexes permit representing both $C$ and $A$ within
the space required to represent $C$ in compressed form, and for example
determine the range $[sp,ep]$ within time $\Oh{|q|\log\sigma}$ and list each
$A[i]$ in time $\Oh{\log^{1+\epsilon} n}$ for any constant $\epsilon>0$
\cite{fmmn2007,nm2007}.

For listing the distinct documents where $q$ appears, one option is to find out
the document to which each $A[i]$ belongs and remove duplicates. This,
however, requires $\Omega(ep-sp+1)$ time; that is, it is proportional to the
{\em total} number of occurrences of $q$, $occ=ep-sp+1$.
This may be much larger than
the number of distinct documents where $q$ appears, $docc$.

Muthukrishnan \cite{Mut02} solved this problem optimally by defining a
so-called {\em document array} $D[1,n]$, so that $D[i]$ is the document suffix
$A[i]$ belongs to. Other required data structures in his solution are an array $C[1,n]$, so
that $C[i] = \max_{j<i} D[j]=D[i]$, and a data structure to compute range
minimum queries on $C$, $RMQ_C(i,j) = \mathrm{argmin}_{i \le k \le j} C[k]$.
Muthukrishnan was able to list all the distinct documents where $q$ appears in
time $\Oh{docc}$ once the interval $A[sp,ep]$ was found. However, the data
structures occupied $\Oh{n\log n}$ bits of space, which is too much if we
consider the compressed self-indexes that solve the basic string search
problem. Another problem is that the resulting documents are not retrieved in
ascending order, which is inconvenient for several purposes.

V\"alim\"aki and M\"akinen \cite{vm2007} were the first to illustrate the
power of wavelet trees for this problem. By representing $D$ with a wavelet
tree, they simulated $C[i] = select_{D[i]}(D,rank_{D[i]}(D,i-1))$ without
storing it. By using a $2n$-bit data structure for $RMQ$ \cite{FH07}, the total
space was reduced to $n\log m (1+o(1)) + \Oh{n}$ bits, and still
Muthukrishnan's algorithm was simulated within reasonable time,
$\Oh{docc \log m}$.

Ranked document retrieval is usually built around two measures: {\em term
frequency}, $\tf_{d,q} = \dfreq(q,d)$ is the number of times the query
$q$ appears in document $d$, and the {\em document frequency} $\df_q$, the
number of different documents where $q$ appears. For example a typical
weighting formula is $w_{d,q} = \tf_{d,q} \times \idf_q$, where $\idf_q =
\log\frac{m}{\df_q}$ is called the {\em inverse document frequency}. Term
frequencies are easily computed with wavelet trees as $\dfreq(q,d) =
rank_d(D,ep)-rank_d(D,sp-1)$. Document frequencies can be computed with just
$2n+o(n)$ more bits for the case of the $D$ array \cite{Sad07}, and on top of
a wavelet tree for the $C$ array for more general scenarios \cite{GNP10}.

In Section~\ref{sec:doclist} we show how our new algorithms solve the document
listing problem within the same time complexity $\Oh{docc \log m}$, without
using any $RMQ$ data structure, while reporting the documents in increasing order.
This is the basis for a novel algorithm to list the documents where two (or
more) queries appear simultaneously. We extend these solutions to temporal and
hierarchical document collections.

\subsection{Inverted Indexes}
\label{sec:invind}

The {\em inverted index} is a classical IR structure \cite{BYRN99,WMB99},
lying at the heart of most modern Web search engines and applications handling
natural-language text collections.
By ``natural language'' texts one refers to those that can be easily split
into a sequence of {\em words}, and where queries are also limited to
words or sequences thereof ({\em phrases}). An inverted index is an array of
{\em lists}. Each array entry corresponds to a different word of the
collection, and its list points to the documents where that word appears.
The set of different words is called the {\em vocabulary}. Compared to the
document retrieval problem for general strings described above, the
restriction of word queries allows inverted indexes to precompute
the answer to each possible word query.

Two main variants of inverted indexes exist \cite{BYMN02,ZM06}. {\em Ranked
retrieval} is aimed at retrieving documents that are most ``relevant'' to a
query, under some criterion. As explained, a popular relevant formula is
$w_{d,q} = \tf_{d,q} \times \idf_q$, but others built on $\tf$ and $\df$,
as well as even more complex ones, have been used.
In inverted indexes for ranked retrieval, the lists point to the documents
where each word appears, storing also the weight of the word in that
document (in the case of $\tf\times\idf$, only $\tf$ values are stored, since
$\idf$ depends only on the word and is stored with the vocabulary).
IR queries are usually formed by various words, so the relevance of the
documents is obtained by some form of combination of the various individual
weights. Algorithms for this type of query have been intensively
studied, as well as different data organizations for this particular task
\cite{PZSD96,WMB99,ZM06,AM06,SC07}. List entries are usually sorted by
{\em descending weights} of the term in the documents.

Ranked retrieval algorithms try to avoid scanning all the involved inverted
lists. A typical scheme is Persin's \cite{PZSD96}. It first retrieves the
shortest list (i.e., with highest $\idf$), which becomes the candidate set,
and then considers progressively longer lists. Only a prefix of the subsequent
lists is considered, where the weights are above a threshold. Those
documents are merged with the candidate set, accumulating relevance values
for the documents that contain both terms. The longer the list, the least
relevant is the term (as the $\tf$s are multiplied by a lower $\idf$), and
thus the shorter the considered prefix of its list.
The threshold provides a time/quality tradeoff.

The second variant is the inverted indexes
for so-called {\em full-text retrieval} (also known as {\em boolean retrieval}). These simply find all
the documents where the query appears. In this case the lists
point to the documents where each term appears, usually in
{\em increasing document} order. Queries can be single words, in which
case the retrieval consists simply of fetching the list of the
word; or disjunctive queries, where one has to fetch the
sorted lists of all the query words and merge them; or
conjunctive queries, where one has to intersect the lists.
Intersection queries are nowadays more popular, as this is Google's default
policy to treat queries of several words.
Another important query where intersection is essential is the phrase query,
where intersecting the documents where the words appear is the first step.

While intersection can be achieved by scanning all the lists in
synchronization, faster approaches aim to exploit the the phenomenon
that some lists are much shorter than others \cite{Zip49}. This
general idea is particularly important when the lists for many terms need to be
intersected.
The amount of recent research on intersection of inverted lists witnesses the
importance of the problem \cite{DM00,BK02,BY04,BYS05,BLOL06,ST07,CM07,BK08}
(see Barbay et al.~\cite{BLOLS09} for a comprehensive survey).
In particular, in-memory algorithms have received much attention lately, as
large main memories and distributed systems make it feasible to hold the
inverted index entirely in RAM.

Needless to say, space is an issue in inverted indexes, especially when
combined with the goal of operating in main memory.
Much research has been carried out on compressing inverted lists
\cite{WMB99,NMNZBY00,ZM06,CM07}, and on the interaction of compression with query
algorithms, including list intersections.
Most of the list compression algorithms for full-text indexes rely on the fact
that the document identifiers are increasing, and that the differences between
consecutive entries are smaller on the longer lists. The differences are thus
represented with encodings that favor small numbers \cite{WMB99}.
Random access is supported by storing sampled absolute values.
For lists sorted by decreasing weights, these techniques can still be adapted:
most documents in a list have small weight values, and
within the same weight one can still sort the documents by increasing
identifier.

A serious problem of the current state of the art is that an IR system usually
must support both types of retrieval: ranked and full-text. For example, this
is necessary in order to provide ranked retrieval on phrases. Yet, to maintain
reasonable space efficiency, the list must be ordered either by decreasing
weights or by increasing document number, but not both. Hence one type of
search will be significantly slower than the other, if affordable at all.

In Section~\ref{sec:invlists} we show that wavelet trees allow one to build a
data structure that permits, within the same space
required for a single compressed inverted index, retrieving the list of
documents of any term in either decreasing-weight or increasing-identifier
order, thus supporting both types of retrieval. Moreover, we can efficiently
support the operations needed to implement any of the intersection algorithms,
namely: retrieve the $i$th element of a list, retrieve the first element
larger than $x$, retrieve the next element, and several more complex ones.
In addition, our structure offers novel ways of carrying out several
operations of interest. These include, among others, the support for stemming
and for structured document retrieval without any extra space cost.

\section{New Algorithms}
\label{sec:alg}

\subsection{Range Quantile}
\label{sec:rqq}

Two n{\"a}ive ways of solving query $\rqq(i,j,k)$ are by sequentially
scanning the range in
time $\Oh{j-i+1}$ \cite{BFP+73}, and storing the answers to the $\Oh{n^3}$
possible queries in a table and returning answers in $\Oh(1)$ time. Neither of these solutions is
really satisfactory.

Until recently there was no work on range quantile queries, but several authors wrote
about {\em range median queries}, the special case in which $k$ is half
the length of the interval between $i$ and $j$. Krizanc et al.~\cite{KMS05}
introduced the problem of preprocessing for range median queries
and gave four solutions, three of which require time superlogarithmic in $n$.
Their fourth solution requires almost quadratic space, storing
$\Oh{n^2 \log \log n / \log n}$ words to answer queries in constant time
(a {\em word} holds $\log\sigma$ bits).
Bose et al.~\cite{BKMT05} considered approximate
queries, and Har-Peled and Muthukrishnan~\cite{HM08} and Gfeller and
Sanders~\cite{GS??} considered batched queries. Recently, Krizanc
et al.'s fourth solution was superseded by one due to Petersen and
Grabowski~\cite{Pet08,PG09}, who slightly reduced the space bound to
$\Oh{n^2 (\log \log n)^2 / \log^2 n}$ words.

At about the same time we presented the early version of our work~\cite{GPT09}, Gfeller and Sanders~\cite{GS??} gave a similar $\Oh{n}$-word data structure that supports range median queries in $\Oh{\log n}$ time and observed in a footnote that ``a generalization to arbitrary ranks will be straightforward''.  A few months later, Brodal and J{\o}rgensen~\cite{BJ09} gave a more involved data structure that still takes $\Oh{n}$ words but only $\Oh{\log n / \log \log n}$ time for queries.  These two papers have now been merged~\cite{BGJS??}.  Very recently, J{\o}rgensen and Larsen~\cite{JL??} proved a matching lower bound for any data structure that takes \(n \log^{\Oh{1}} n\) space.

In the sequel we show that, if $S$ is represented using a wavelet tree,
we can answer general range quantile queries in $\Oh{\log u}$ time,
where $u \le \min(\sigma,n)$ is the number of distinct symbols in $S$. As
explained in Section~\ref{sec:wt}, within these $n\log\sigma+\Oh{n}$ bits of
space we can also retrieve any element $S[i]$ in time $\Oh{\log u}$, so our
data structure actually {\em replaces} $S$ (requiring only $\Oh{n}$ extra
bits). The latest alternative structure~\cite{JL??} may achieve slightly
better time but it requires $\Oh{n\log n}$ extra bits of space, apart from
being significantly more involved.

\begin{thrm}
\label{thm:quantile}
Given a sequence $S[1,n]$ storing $u$ distinct values over alphabet
$[1,\sigma]$, we can represent $S$ within $n\log\sigma + \Oh{n}$ bits, so
that range quantile queries are solved in time $\Oh{\log u}$. Within that
time we can also know the number of times the returned value appears in the
range.
\end{thrm}
\begin{proof}
We represent $S$ using a wavelet tree $T$, as in Lemma~\ref{lem:wt}.
Query $\rqq(i,j,k)$ is then solved as follows.
We start at the root of $T$ and consider its bitmap $B_{root}$.
We compute $n_l = rank_0(B_{root},j)-rank_0(B_{root},i-1)$, the number of 0s in
$B_{root}[i,j]$. If $n_l \ge k$, then there are at least $k$ symbols in $S[i,j]$
that label leaves descending from the left child $T_l$ of $T$, and thus we
must find the $k$th symbol on $T_l$. Therefore we continue recursively on $T_l$
with the new values $i \leftarrow rank_0(B_{root},i-1)+1$,
$j \leftarrow rank_0(B_{root},j)$, and $k$ unchanged. Otherwise, we must
descend to the right child, mapping the range to
$i \leftarrow rank_1(B_{root},i-1)+1$ and $j \leftarrow rank_1(B_{root},j)$.
In this case, since we have discarded $n_l$ numbers that are already to the
left of the $k$th value, we set $k \leftarrow k-n_l$.
When we reach a leaf, we just return its label. Furthermore, we have that
the value occurs $j-i+1$ times in the original range.
Since $T$ is balanced and we spend constant time at each node as we descend,
our search takes $\Oh{\log u}$ time.
\qed
\end{proof}


\begin{algorithm}[t]
\caption{New wavelet tree algorithms:
$\mathbf{rqq}(v_{root},i,j,k)$ returns $(\rqq(S,i,j,k),f)$ on the wavelet tree
of sequence $S$, assuming $k \le j-i+1$, and where $f$ is the frequency of
the returned element in $S[i,j]$;
$\mathbf{rnv}(v_{root},i,j,0,x)$ returns $(\rnv(S,i,j,x),f,p)$, where $f$ is the
frequency and $p$ is the smallest rank of the returned element in the
multiset $S[i,j]$ (the element is $\perp$ if no answer exists); and
$\mathbf{rint}(v_{root},i_1,j_1,i_2,j_2,[y^s,y^e])$ solves an extension of query
$\rint(S,i_1,j_1,i_2,j_2)$ outputting triples $(y,f_1,f_2)$, where $y$ are the
common elements, $f_1$ is their frequency in $S[i_1,j_1]$, and $f_2$ is their
frequency, in $S[i_2,j_2]$, and moreover $y^s \le y \le y^e$.}
\label{alg:rqqrnvint}
\begin{tabular}{ccc}
\begin{minipage}{0.32\textwidth}
$\mathbf{rqq}(v,i,j,k)$
\begin{algorithmic}
\IF{$v$ is a leaf}
	\RET $(label(v),j-i+1)$
\ELSE
	\STATE $i_l \leftarrow rank_0(B_v,i-1)+1$
	\STATE $j_l \leftarrow rank_0(B_v,j)$
	\STATE $i_r \leftarrow i-i_l$, $j_r \leftarrow j-j_r$
	\STATE $n_l \leftarrow j_l-i_l+1$
	\IF{$k \le n_l$}
		\RET $\mathbf{rqq}(v_l,i_l,j_l,k)$
	\ELSE 	\RET $\mathbf{rqq}(v_r,i_r,j_r,k-n_l)$
	\ENDIF
\ENDIF
\STATE
\STATE
\STATE
\STATE
\STATE
\STATE
\STATE
\STATE
\end{algorithmic}
\end{minipage}
&
\begin{minipage}{0.35\textwidth}
$\mathbf{rnv}(v,i,j,p,x)$
\begin{algorithmic}
\IF{$i>j$} \RET $(\perp,0,0)$
\ELSIF{$v$ is a leaf}
	\RET $(x,j-i+1,p)$
\ELSE
	\STATE $i_l \leftarrow rank_0(B_v,i-1)+1$
	\STATE $j_l \leftarrow rank_0(B_v,j)$
	\STATE $i_r \leftarrow i-i_l$, $j_r \leftarrow j-j_r$
	\STATE $n_l \leftarrow j_l-i_l+1$
	\IF{$x \in labels(v_r)$} \RET $\mathbf{rnv}(v_r,i_r,j_r,p+n_l,x)$
	\ELSE
		\STATE $(y,f) \leftarrow \mathbf{rnv}(v_l,i_l,j_l,p,x)$
		\IF{$y \not= \perp$} \RET $(y,f)$
		\ELSE \RET $\mathbf{rnv}(v_r,i_r,j_r,p+n_l,$
		      \STATE \hspace{1.8cm} $\min labels(v_r))$
		\ENDIF
	\ENDIF
\ENDIF
\end{algorithmic}
\end{minipage}
&
\begin{minipage}{0.32\textwidth}
$\mathbf{rint}(v,i_1,j_1,i_2,j_2,rng)$
\up
\begin{algorithmic}
\IF{$i_1 > j_1 ~\lor~ i_2 > j_2$}
	\RET
\ELSIF{$labels(v) \cap rng = \emptyset$}
	\RET
\ELSIF{$v$ is a leaf}
	\OUTPUT $(label(v),$
	\STATE \hspace{1.2cm} $j_1-i_1+1,j_2-i_2+1)$
\ELSE
	\STATE $i_{1l} \leftarrow rank_0(B_v,i_1-1)+1$
	\STATE $j_{1l} \leftarrow rank_0(B_v,j_1)$
	\STATE $i_{1r} \leftarrow i_1-i_{1l}$, $j_{1r} \leftarrow j_1-j_{1r}$
	\STATE $i_{2l} \leftarrow rank_0(B_v,i_2-1)+1$
	\STATE $j_{2l} \leftarrow rank_0(B_v,j_2)$
	\STATE $i_{2r} \leftarrow i_2-i_{2l}$, $j_{2r} \leftarrow j_2-j_{2r}$
	\STATE $\mathbf{rint}(v_l,i_{1l},j_{1l},i_{2l},j_{2l},rng)$
	\STATE $\mathbf{rint}(v_r,i_{1r},j_{1r},i_{2r},j_{2r},rng)$
\ENDIF
\STATE
\STATE
\STATE
\STATE
\end{algorithmic}
\end{minipage}
\end{tabular}
\end{algorithm}

Algorithm~\ref{alg:rqqrnvint} (left) gives pseudocode.
Note that, if $u$ is constant, then so is our query time.

\subsection{Range Next Value}
\label{sec:rnv}

Again, two naive ways of solving query $\rnv(i,j,x)$ on
sequence $S[1,n]$ are scanning in $\Oh{j-i+1}$ worst-case time, and
precomputing all the possible answers in $\Oh{n^3}$ space to achieve constant
time queries. Crochemore et al.~\cite{CIR07} reduced the space to $\Oh{n^2}$ words while
preserving the constant query time. Later, Crochemore et al.~\cite{CIKRW08}
further improved the space to $\Oh{n^{1+\epsilon}}$ words. Alternatively,
M\"akinen et al.~\cite[Lemma 4]{MNU05} give a simple $\Oh{n}$-words space
solution based on an augmented binary search tree. This yields time
$\Oh{\log u}$, where once again $u \le \min(n,\sigma)$ is the number of distinct
symbols in $S$ and $[1,\sigma]$ the domain of values. For the particular case
of semi-infinite queries (i.e., $i=1$ or $j=n$) one can use an $\Oh{n}$-words
and $\Oh{\log\log n}$ time solution by Gabow et al.~\cite{GBT84}.

By using wavelet trees, we also solve the general problem in time $\Oh{\log u}$.
Our space is better than the simple linear-space solution,
$n+\Oh{n/\log\sigma}$ words ($n$ of which actually replace the sequence).


\begin{thrm}
\label{thm:nextvalue}
Given a sequence $S[1,n]$ storing $u$ distinct values over alphabet
$[1,\sigma]$, we can represent $S$ within $n\log\sigma + \Oh{n}$ bits, so
that range next value queries are solved in time $\Oh{\log u}$. Within the
same time we can return the position of the first occurrence of the value in
the range.
\end{thrm}
\begin{proof}
We represent $S$ using a wavelet tree $T$, as in Lemma~\ref{lem:wt}.
Query $\rnv(i,j,x)$ is then solved as follows.
We start at the root of $T$ and consider its bitmap $B_{root}$.
If $x$ labels a leaf descending by the right child $T_r$, then the left
subtree is irrelevant and we continue recursively on $T_r$, with
the new values $i \leftarrow rank_1(B_{root},i-1)+1$ and
$j \leftarrow rank_1(B_{root},j)$. Otherwise, we must descend to the left
child $T_l$, mapping the range to $i \leftarrow rank_0(B_{root},i-1)+1$ and
$j \leftarrow rank_0(B_{root},j)$. If our interval $[i,j]$ becomes empty at
any point, we return with no value.

When the recursion returns from $T_r$ with no value, we return no value as
well. When it returns from $T_l$ with no value, however, there is still a
chance that a number $\ge x$ appears on the right in the interval $[i,j]$.
Indeed, if we descend to $T_r$ and map $i$ and $j$ accordingly, and the
interval is not empty, then we want the minimum value of that interval,
that is, the minimum value in $S_l[i,j]$. This is a particular case of a
$\rqq$ query carried out on a wavelet (sub)tree $T_r$.
The overall time is $\Oh{\log u}$.
\qed
\end{proof}

Algorithm~\ref{alg:rqqrnvint} (middle) gives pseudocode.
While our space gain may not appear very impressive, we point out that our solution
requires only $\Oh{n}$ extra bits on top of the sequence (if we accept the
logarithmic slowdown in accessing $S$ via the wavelet tree). Moreover,
we can use the same wavelet tree to carry out the other algorithms, instead of
requiring a different data structure for each. This will be relevant for the
applications, which need support for several of the operations simultaneously.

\subsection{Range Intersection}
\label{sec:int}

The query $\rint(i_1,j_1,i_2,j_2)$, which finds the common symbols in
two ranges of a sequence $S[1,n]$ over alphabet $[1,\sigma]$, appears naturally
in many cases. In particular, a simplified variant where the two ranges to
intersect are sorted in increasing order arises when intersecting full-text
inverted lists, when solving intersection, phrase, or proximity queries.

Worst-case complexity measures depending only on the range sizes are of little
interest for this problem, as an adversary can always force us to completely
traverse both ranges, and time complexity $\Oh{j_1-i_1+j_2-i_2+1}$ is easily
achieved through merging\footnote{If the ranges are already ordered; otherwise
a previous sorting is necessary.}.
More interesting are {\em adaptive} complexity measures, which
define a finer {\em difficulty measure} for problem instances. For example,
in the case of sorted ranges, an instance where the first element of the second
range is larger than the last element of the first range is easier (one can
establish the emptiness of the result with just one well-chosen comparison)
than another where elements are mixed.

A popular measure for this case is called {\em alternation} and noted
$\alpha$ \cite{BK02}. For two sorted sequences without repetitions, $\alpha$ can
be defined as the number of switches from one sequence to the other in the
sorted union of the two ranges, or equivalently, as the time complexity of a
nondeterministic program that guesses which comparisons to carry out, or
equivalently as the length of a certificate that, through the results of
comparing elements of both sequences, is sufficient to prove what the result
is. This definition can be extended to intersecting $k$ ranges. Formally, the
measure $\alpha$ is defined through a function $C:[1,\sigma]\rightarrow [0,k]$,
where $C[c]$ gives the number of any range where symbol $c$ does not appear,
and $C[c]=0$ if $c$ appears in all ranges. Then $\alpha$ is the number of zeros
in $C$ plus the minimum possible number of switches (i.e., $C[c] \not= C[c+1]$)
in such a function. A lower bound in terms of alternation (still holding for
randomized algorithms) \cite{BK02} is
$\Omega\left(\alpha \cdot \sum_{1\le r\le k} \log\frac{n_r}{\alpha}\right)$,
where $n_r$ is the length of the $r$th range. There exist adaptive algorithms
matching this lower bound \cite{DM00,BK02,BK08}.

We show now that the wavelet tree representation of $S[1,n]$ allows a rather
simple intersection algorithm that approaches the lower bound, even if one
starts from ranges of {\em disordered} values, possibly with repetitions.
For $k=2$, we start from both ranges $[i_1,j_1]$ and $[i_2,j_2]$ at the root
of the wavelet tree. If either range is empty, we stop. Otherwise we map both
ranges to the left child of the root using $rank_0$, and to the right child
using $rank_1$. We continue recursively on the branches where both intervals
are nonempty. If we reach a leaf, then its corresponding symbol is in the
intersection, and we know that there are $j_1-i_1+1$ copies of the symbol
in the first range, and $j_2-i_2+1$ in the second. For $k$ ranges $[i_r,j_r]$,
we maintain them all at each step, and abandon a path as soon as any of the
$k$ ranges becomes empty. Algorithm~\ref{alg:rqqrnvint} (right) gives
pseudocode for the case $k=2$.

\begin{lemma} \label{lem:adaptive}
The algorithm just described requires time
$\Oh{\alpha k\log\frac{u}{\alpha}}$, where $u$ is the number of distinct
values in the sequence and $\alpha$ is the alternation complexity of the
problem.
\end{lemma}
\begin{proof}
Consider the function $p : \Sigma \rightarrow \{0,1\}^*$, so that $p(c)$ is a
bit stream of length equal to the depth of the leaf representing symbol $c$ in
the wavelet tree. More precisely, $p[i]$ is 0 if the leaf descends from the
left child of its ancestor at depth $i$, and 1 otherwise. That is, $p(c)$
describes the path from the root to the wavelet tree leaf labeled $c$.

Now let $T_r$ be the {\em trie} (or digital tree) formed by the strings
$p(c)$ for all those $c$ appearing in $S[i_r,j_r]$, and let $T_\cap$ be
the trie formed by the branches present in all $T_r$, $1\le r\le k$.
It is easy to see that $T_\cap$ contains precisely the wavelet tree nodes
traversed by our intersection algorithm, so the complexity of our algorithm is
$\Oh{|T_\cap|}$.

We show now that $|T_\cap|$ has at most $\alpha$ leaves. The leaves of $T_\cap$
that are wavelet tree leaves correspond to the symbols that belong to the
intersection, and thus to the number of 0s in any function $C$. This is
accounted for in measure $\alpha$. So let us focus on the other leaves of
$T_\cap$. Consider two consecutive leaves of $T_\cap$ that are not wavelet
tree leaves $u_1$ and $u_2$, and any symbols $c_1<c_2$ whose wavelet tree
leaves $v_1$ and $v_2$ descend from $u_1$ and $u_2$, respectively. If there
were a single range $S[i_r,j_r]$ where $c_1$ and $c_2$ would not belong, then
the lowest common ancestor of $v_1$ and $v_2$ would not
belong to $T_\cap$, and thus there could not be two leaves $u_1$ and $u_2$ in
$T_\cap$. Therefore, for each pair of consecutive leaves in $T_\cap$ there is
at least one switch in $C$, and thus there are at most $\alpha$ leaves in
$T_\cap$. Thus, by Lemma~\ref{lem:nodeset}, $T_\cap$ has
$\Oh{\alpha\log\frac{u}{\alpha}}$ nodes. To obtain the final cost we
multiply by $k$, which is the cost of maintaining the $k$ ranges throughout the
traversal.
\qed
\end{proof}

In the case where all the lists are sorted and without repetitions (so $n_r
\le u$), our algorithm complexity is pretty close to the lower bound, matched
when all $n_r = u$.
Note also that our algorithm is easily extended to handle the so-called
{\em $(t,k)$-thresholded problem} \cite{BK02}, where we return any symbol
appearing in at least $t$ of the $k$ ranges. It is simply a matter of abandoning a
range only when more than $k-t$ ranges have become empty.

A different form of carrying out the intersection is via the query
$\rnv(S,i,j,x)$: Start with $x_1 \leftarrow \rnv(S,i_1,j_1,1)$ and $x_2
\leftarrow \rnv(S,i_2,j_2,x_1)$. If $x_2 > x_1$ then continue with
$x_1 \leftarrow \rnv(S,i_1,j_1,x_2)$; if now $x_1 > x_2$ then continue with
$x_2 \leftarrow \rnv(S,i_2,j_2,x_1)$; and so on. If at any moment $x_1 = x_2$
then output it as part of the intersection and continue with
$x_1 \leftarrow \rnv(S,i_1,j_1,x_2+1)$. It is not hard to see that there must
be a switch in $C$ for each step we carry out, and therefore the cost is
$\Oh{\alpha\log u}$.

To reduce the cost to $\Oh{\alpha\log\frac{u}{\alpha}}$,
we carry out a {\em fingered search} in $\rnv$ queries, that is, we remember
the path traversed from the last time we called $\rnv(S,i,j,x)$ and only
retraverse the necessary part upon calling $\rnv(S,i,j,x')$ for $x'>x$. For
this reason we move upwards from the leaf where the query for $x$ was solved
until reaching the first node $v$ such that $x' \in labels(v)$, and complete
the $\mathbf{rnv}$ procedure from that node. Since the total work done by this point is
proportional to the number of distinct ancestors of the $\alpha$ leaves arrived
at, the complexity is $\Oh{\alpha\log\frac{u}{\alpha}}$ by
Lemma~\ref{lem:nodeset}.

This second procedure is the basis of most algorithms for intersecting two or
more lists \cite{BLOLS09}. The $\mathbf{rint}$ method we have presented is
simpler, potentially faster, and more flexible (e.g., it is easily adapted to
$t$-thresholded queries). Moreover, it is specific to the wavelet tree.

\section{Document Listing and Intersections}
\label{sec:doclist}

%

The algorithm for $\rr(P,x^s,x^e,y^s,y^e)$ queries described in
Section~\ref{sec:wt} can be used to solve problem $\dlist(q)$, as follows. As
explained in Section~\ref{sec:basicdoc}, use a (compressed) suffix array $A$
to find the range $A[sp,ep]$ corresponding to query $q$, and use a wavelet tree
on the document array $D[1,n]$ on alphabet $[1,m]$, so that the answer is the
set of distinct document numbers $d_1 < d_2 < \ldots < d_{docc}$ in $D[sp,ep]$.
Then $\rr(D,sp,ep,1,m)$ returns the $docc$ document numbers, in
order, in total time $\Oh{docc \log\frac{m}{docc}}$. Moreover, procedure
$\mathbf{report}$ in Algorithm~\ref{alg:range} also retrieves the frequencies
of each $d_i$ in $D[sp,ep]$, outputting the pairs $(d_i,\tf_{q,d_i})$ within
the same cost. (As explained, arbitrary frequencies $\tf_{d,q}=\dfreq(q,d)$
can also be obtained in time $\Oh{\log m}$ by two $rank_d$ queries on $D$.)
Alternative solutions using $\rqq$ or $\rnv$ queries are possible, and will be
explored later for other applications.

As explained in Section~\ref{sec:basicdoc}, this is simpler and requires less
space than various previous solutions\footnote{It is even better than our
previous solution based on $\rqq$ queries \cite{GPT09}, which takes time
$\Oh{docc \log m}$.}, and has the additional benefit of delivering the
documents in increasing document identifier order. This enables us to
extend the algorithm to more complex scenarios, as shown in
Section~\ref{sec:invlists}.

Now consider $k$ queries $q_1, q_2, \ldots, q_k$, and the problem of listing
the documents where all those queries appear (i.e.,
problem $\dint(q_1,\ldots,q_k)$). With the suffix array we can map the
queries to ranges $[sp_r,ep_r]$, and then the problem is that of finding the distinct document numbers that appear in all those ranges. This corresponds
exactly to query $\rint(D,sp_1,ep_1,\ldots,sp_k,ep_k)$,
which we have solved in Section~\ref{sec:int}. We
have indeed solved a more general variant where we list the documents (and
their $\tf_{d,q_r}$ values) where at least $t$ of the $k$ terms appear. Note
this corresponds to the disjunctive query for the case $t=1$.

\subsection{Temporal and Hierarchical Documents}
\label{sec:temphier}

The simplest extension when we have versioned or hierarchical documents is to
restrict queries $\dlist(q)$ and $\dint(q_1,\ldots,q_k)$ to a range of documents
$[d_{min},d_{max}]$, which represents a temporal interval or a subtree of the
hierarchy in which we are interested. Such a restricted document listing and
intersection is easily supported by setting $rng = [d_{min},d_{max}]$ in
procedures $\mathbf{report}$ (Algorithm~\ref{alg:range}) and $\mathbf{rint}$
(Algorithm~\ref{alg:rqqrnvint}), respectively. The complexities are
$\Oh{docc\log\frac{d_{max}-d_{min}+1}{docc}}$ for listing and
$\Oh{\alpha\log\frac{d_{max}-d_{min}+1}{\alpha}}$ for intersections, due to
Lemma~\ref{lem:leafrangeset}.

When the hierarchical documents represent nodes in an XML collection, other
queries of interest become obvious. Indeed, how to carry out ranking on XML
collections is an unresolved issue, with very complex ranking proposals
counterweighted by others advocating simple measures. Rather than trying
to cover such a broad topic, we refer the reader to comprehensive surveys and
discussions in the article by Hiemstra and Mihajlovi\'c \cite{HM05}, the
PhD thesis of Pehcevski \cite[Ch.~2]{Pec06}, and the recent book by Lalmas
\cite[Ch.~6]{Lal09}.

In most models, the frequency of a term within a subtree,
and the size of such subtree, are central to the definition of ranking
strategies. The latter is usually easy to compute from the sequence
representation. The former, a generalization of $\dfreq$ to ranges, can
actually be computed with query $\rc(D,sp,ep,dl,dr)$,
defined in the Introduction (see also Algorithm~\ref{alg:range}),
where $[sp,ep]$ is the suffix array range
corresponding to query $q$, and $[dl,dr]$ is the range of documents
corresponding to our structural element. This query also takes time
$\Oh{\log m}$.

\subsection{Restricting to Retrievable Units}
\label{sec:units}

We focus now on a more complex issue that is also essential for XML ranked
retrieval. Query languages such as XPath and XQuery define structural
constraints together with terms of interest. For example, one might wish to
retrieve books about the term ``cryptography'', or rather book sections about
that term, in each case ranked by the relevance of the term. Thus the
definition of the {\em retrievable unit} (books, sections) comes in the query
together with the terms (cryptography) whose relevance is to be computed with
respect to the retrievable units that contain it.
We show now how to support a simple model where the retrievable units are
defined by an XML {\em tag} name, and consider other models at the end.
We report the smallest retrievable unit containing the query occurrences.

Following common models of XML data (e.g., \cite{ACMNNSV10}), we consider that
text data can appear only at the leaves of the XML structure, so that we
create extra leaves if text data appears between consecutive structural
elements (a bitmap may be used to mark leaves that do not contain any text
data, but we omit this detail here for clarity). Thus, each leaf of the XML
tree will be associated with a document number, 1 to $m$, so that $d_i$ will be the
document associated to the $i$th leaf. The XML tree, containing $n$ nodes, will
be represented using a sequence $P[1,2n]$ of parentheses \cite{Jac89}.
These are obtained through a preorder traversal, by appending an opening
parenthesis when we reach a node and a closing one when we leave it. A tree
node will be identified with the position of its opening parenthesis in $P$.
Several succinct data structures can represent the parentheses within $2n+o(n)$
bits and simulate a wealth of tree operations in constant time
(e.g., \cite{SN10}).

In addition we represent a sequence $\Tag[1,2n]$ giving
the tag name associated to each parenthesis in $P$. Sequence $\Tag$ is
represented using a wavelet tree in $2n\log \tau + \Oh{n}$ bits of space,
where $\tau$ is the number of distinct tags in the collection. Finally, for
each distinct tag name $t$ we store a parenthesis representation $P_t$ of the
nodes of the XML tree that are tagged $t$. The total space for $P$, $\Tag$,
and the $P_t$ trees is $2n\log\tau + \Oh{n}$ bits.

A first task we can carry out is, given an occurrence in document number
(i.e., leaf) $i$, find $\expand(t,i)$, the range of documents (i.e., leaves)
corresponding to its lowest ancestor tagged $t$. This allows
us to find the closest retrievable unit to which the occurrence at leaf $i$
must be assigned. We use operation $j = \mathit{selectLeaf}(P,i)$ to find the
$i$th leaf of $P$. Then $r = rank_t(\Tag,j)$ finds the rank of the last
occurrence of $t$ in $\Tag$ preceding $j$. If $P_t[r] = \mathtt{'('}$, then
$r$ is the lowest ancestor of $i$ tagged $t$, otherwise it is $r \leftarrow
\textit{parent}(P_t,r)$, the node tagged $t$ that encloses position $r$.
Finally, position $r$ is mapped back to the global tree $P$ with
$p = select_t(\Tag,r)$, and we return the range of leaves corresponding to $p$,
$\expand(t,i) = \mathit{leaf\_range}(p) = [\mathit{rankLeaf}(P,p)+1,
\mathit{rankLeaf}(P,p+2\cdot \mathit{subtreeSize}(P,p))]$,
where {\em rankLeaf} and {\em subtreeSize} are self-explanatory tree operations.
The process takes $\Oh{\log \tau}$ time, dominated by the costs to operate on
$\Tag$. Algorithm~\ref{alg:hier} (left) gives pseudocode.

\begin{algorithm}[t]
\caption{Algorithms for hierarchical document listing and intersections:
$\mathbf{exp}(\Tag,P,P_t,t,i)$ computes the node in $P$ for
$\mathit{expand}(t,i)$ and $\mathbf{leafRange}(P,p)$ computes
$\mathit{leaf\_range}(p)$; $\mathbf{hdfreq}(P,D,sp,ep,p)$ computes the
frequency of $p$ in $D[sp,ep]$; $\mathbf{hdlist}(A,D,\Tag,P,P_t,t,q)$ lists
the retrievable units where $q$ appears; and
$\mathbf{hdint}(\Tag,P,P_t,t,v_{root},i_1,j_1,i_2,j_2,rng)$ lists the
retrievable units with leaves in both $D[i_1,j_1]$ and $D[i_2,j_2]$,
with their frequencies in both ranges, and subject to belonging to document
range $rng$ (which is assumed not to split any retrievable unit).}
\label{alg:hier}
\begin{tabular}{cc}
\begin{minipage}{0.34\textwidth}
$\mathbf{exp}(\Tag,P,P_t,t,i)$
\home{\up}
\begin{algorithmic}
\STATE $j \leftarrow \mathit{selectLeaf}(P,i)$
\STATE $r \leftarrow \mathbf{rank}(\Tag,t,j)$
\IF{$P_t[r] = \mathtt{')'}$}
	\STATE $r \leftarrow parent(P_t,r)$
\ENDIF
\RET $\mathbf{select}(\Tag,t,r)$
\STATE
\end{algorithmic}

$\mathbf{leafRange}(P,p)$
\home{\up}
\begin{algorithmic}
\RET $[\mathit{rankLeaf}(P,p)+1,$
\STATE  $\mathit{rankLeaf}(P,p+2\cdot \mathit{subtreeSize}(P,p))]$
\STATE
\end{algorithmic}

$\mathbf{hdfreq}(P,D,sp,ep,p)$
\home{\up}
\begin{algorithmic}
\STATE $[dl,dr] \leftarrow \mathbf{leafRange}(P,p)$
\RET $\mathbf{count}(D,sp,ep,[dl,dr])$
\STATE
\end{algorithmic}

$\mathbf{hdlist}(A,D,\Tag,P,P_t,t,q,[d_{min},d_{max}])$
\up
\begin{algorithmic}
\STATE $[sp,ep] \leftarrow pattern\_search(A,q)$
\STATE $v \leftarrow root(D)$
\STATE $(d,f,r) \leftarrow \mathbf{rnv}(v,sp,ep,0,d_{min})$
\WHILE{$d \not= \perp ~\land~ d \le d_{max}$}
	\STATE $p \leftarrow \mathbf{exp}(\Tag,P,P_t,t,d)$
	\OUTPUT $p$
	\STATE $[dl,dr] \leftarrow \mathbf{leafRange}(P,p)$
	\STATE $(d,f,r) \leftarrow \mathbf{rnv}(v,sp,ep,0,dr+1)$
\ENDWHILE
\end{algorithmic}
\end{minipage}
&
\begin{minipage}{0.66\textwidth}
$\mathbf{hdint}(\Tag,P,P_t,t,v,i_1,j_1,i_2,j_2,rng)$
\begin{algorithmic}
\IF{$i_1 > j_1 ~\lor~ i_2 > j_2 ~\lor~ rng = \emptyset$}
	\RET
\ELSIF{$v$ is a leaf}
	\OUTPUT $(label(v),j_1-i_1+1,j_2-i_2+1)$
\ELSE
	\STATE $[dl,dr] \leftarrow \emptyset$,
	       $f_1 \leftarrow 0$, $f_2 \leftarrow 0$
	\STATE $i_{1l} \leftarrow rank_0(B_v,i_1-1)+1$,
	       $j_{1l} \leftarrow rank_0(B_v,j_1)$
	\STATE $i_{1r} \leftarrow i_1-i_{1l}$, $j_{1r} \leftarrow j_1-j_{1r}$
	\STATE $i_{2l} \leftarrow rank_0(B_v,i_2-1)+1$,
	       $j_{2l} \leftarrow rank_0(B_v,j_2)$
	\STATE $i_{2r} \leftarrow i_2-i_{2l}$, $j_{2r} \leftarrow j_2-j_{2r}$
	\IF{$i_{1l} \le j_{1l} ~\land~ i_{2l} \le j_{2l} ~\land~
	    i_{1r} \le j_{1r} ~\land~ i_{2r} \le j_{2r} ~\land~~~~~~~~~~~~$
	    \mbox{$~~~~labels(v_l) \cap rng \not= \emptyset ~\land~
labels(v_r) \cap rng \not= \emptyset$}}
	   \STATE $p_l \leftarrow \mathbf{exp}(\Tag,P,P_t,t,\max labels(v_l))$
	   \STATE $p_r \leftarrow \mathbf{exp}(\Tag,P,P_t,t,\min labels(v_r))$
	   \IF{$p_l = p_r$}
	      \STATE $[dl,dr] \leftarrow \mathbf{leafRange}(P,p_l)$
	      \STATE $f_1 \leftarrow \mathbf{count}(v,i_1,j_1,[dl,dr])$
	      \STATE $f_2 \leftarrow \mathbf{count}(v,i_2,j_2,[dl,dr])$
	   \ENDIF
	\ENDIF
	\STATE $\mathbf{hdint}(\Tag,P,P_t,t,v_l,i_{1l},j_{1l},i_{2l},j_{2l},
				(labels(v_l) \cap rng) - [dl,dr])$
	\IF{$f_1 > 0 ~\land~ f_2 > 0$}
	   	\OUTPUT $(p_l,f_1,f_2)$
	\ENDIF
	\STATE $\mathbf{hdint}(\Tag,P,P_t,t,v_r,i_{1r},j_{1r},i_{2r},j_{2r},
				(labels(v_r) \cap rng) - [dl,dr])$
\ENDIF
\end{algorithmic}
\end{minipage}
\end{tabular}
\end{algorithm}

If we now want to count the number of occurrences of our query $q$ in a
retrievable node $p$, we need to count the number of occurrences of the range
of leaves (i.e., document numbers) below $p$ within the interval $D[sp,ep]$
corresponding to query $q$. Such a range is easily obtained in constant time as
$[dl,dr] = \mathit{leaf\_range}(p)$. Then the result is
$\rc(D,sp,ep,dl,dr)$, as explained.


To carry out document listing restricted to structural elements tagged $t$, we
build on range next value queries. We start with $d_1 = \rnv(D,sp,ep,1)$,
which gives us the smallest (leaf) document number in $D[sp,ep]$. Now we
compute $[dl_1,dr_1]=\expand(t,d_1)$, the range of the lowest node tagged $t$
that contains $d_1$. Then we find the next leaf document using
$d_2 = \rnv(D,sp,ep,dr_1+1)$, and so on. In general,
$d_{i+1} = \rnv(D,sp,ep,dr_i+1)$. Algorithm~\ref{alg:hier} (left) gives
pseudocode. The cost per document retrieved is $\Oh{\log\tau + \log m}$.
However, using the fingered search on $\mathbf{rnv}$ outlined in
Section~\ref{sec:int}, the overall cost reduces to
$\Oh{docc \left(\log\tau + \log\frac{m}{docc}\right)}$.
If we wish to additionally restrict the retrieval to
documents in the range $[d_{min},d_{max}]$, we simply start with
$d_1 = \rnv(D,sp,ep,d_{min})$ and stop when we retrieve a document larger than
$d_{max}$. The cost improves to
$\Oh{docc \left(\log\tau + \log\frac{d_{max}-d_{min}+1}{docc}\right)}$ due to
Lemma~\ref{lem:leafrangeset}.
Complexity returns to $\Oh{docc \log m}$ if we compute also the frequency in
each retrievable unit using $\mathbf{hdfreq}$.

Finally, to carry out intersections restricted to retrievable units,
we follow in principle the same algorithm outlined in
Section~\ref{sec:int}. The difference is that we must not split retrievable
ranges. Therefore, when we are at any wavelet tree node, before going to the
left child that represents the range of symbols $[d_a,d_b]$ and/or to the
right child representing range $[d_b+1,d_c]$, we first find out whether there
is a retrievable unit covering $[d_b,d_b+1]$. To do this we compute
$\expand(t,d_b) = [dl,dr]$ and $\expand(t,d_b+1) = [dl',dr']$. Since the leaves
are consecutive, there are only two possibilities: either $[dl,dr] = [dl',dr']$
or they are disjoint (and $dr=d_b$ and $dl' = d_b+1$). In the latter case we
proceed with the recursion as in Section~\ref{sec:int}. In the former case, we
descend to the left child with document range restricted to $[d_a,dl-1]$, then
report document $[dl,dr]$ if it belongs to the intersection, and finally
descend to the right child with document range restricted to $[dr+1,d_c]$. By
``descending with document range restricted to $[x,y]$'' we mean we abandon
branches whose document range has no intersection with $[x,y]$, and such
restrictions are inherited as we descend.

Algorithm~\ref{alg:hier} (right) gives pseudocode.
The complexity is the same as if the retrievable units were materialized into
consecutive document numbers, that is, $\Oh{\alpha \log\frac{m}{\alpha}}$ under
this interpretation. The only extra cost is the computation of $f_1$ and $f_2$.
Note, however, that these are computed with a {\em \rc} query
restricted to the local subtree, and thus the cost at height $h$ is $\Oh{h}$.
Moreover, this is computed only for nodes of trie $T_\cap$ (recall
Lemma~\ref{lem:adaptive}) having two children, that is, at most $\alpha$ times.
The most expensive case is thus when all those $\alpha$ nodes are as high as
possible in the wavelet tree, in which case the $\Oh{h}$ costs add up to
$\Oh{\alpha\log\frac{m}{\alpha}}$ and do not affect the complexity. Once
again, we can restrict the results to a range $[d_{min},d_{max}]$ with the
usual time improvement.

Other possibilities for marking the retrievable documents can be supported, as
long as one is able to find the lowest retrievable ancestor of any leaf. For
example we could mark retrievable nodes in a bitmap $B[1,2n]$ aligned with
$P$, where we set to 1 the opening and closing parentheses of retrievable
nodes. Then we can compute $\expand(B,i)$ via $rank$ and $select$ operations on
$B$ in constant as follows. We start with $j = \mathit{selectLeaf}(P,i)$, then
$p = select_1(rank_1(B,j))$, then if $P[p] = \mathtt{')'}$ we recompute
$p=parent(P,p)$, and finally $\expand(B,i) = \mathit{leaf\_range}(p)$.
In cases where the retrievable units are defined dynamically, say from previous
parts of the query processing, we can store them in a balanced tree, so that
query $select_1(rank_1(B,j))$ (which is actually a predecessor query) can be
answered in $\Oh{\log n}$ time.

\section{Inverted Lists}
\label{sec:invlists}

Recall $m$ is the total number of documents in the collection and let $\nu$ be the number
of different terms. Let $L_t[1,\df_t]$ be the list of document identifiers where
term $t$ appears, in decreasing weight order (for concreteness we will assume
we store $\tf$ values in the lists as weights, but any weight will do).
Let $n = \sum_t \df_t$ be the total number of occurrences of {\em
distinct} terms in the documents, and $N = \sum_{t,d} \tf_{t,d}$ the total
length, in words, of the text collection (thus $m \le n \le \min(m\nu,N)$).
Finally, let $|q|$ be the number of terms in query $q$.

We propose to concatenate all the lists $L_t$ into a unique sequence $L[1,n]$,
and store for each term $t$ the starting position $s_t$ of list $L_t$ within
$L$. The sequence $L$ of document identifiers is then represented with a wavelet
tree.

The $\tf$ values themselves are stored in differential and run-length
compressed form in a separate sequence.
More precisely, we mark the $v_t$ different $\tf_{t,d}$ values of each list in
a bitmap $T_t[1,m_t]$, where $m_t = \max_d \tf_{t,d}$, and the $v_t$ points in
$L_t[1,\df_t]$ where value $\tf_{d,t}$ changes, in a bitmap $R_t[1,\df_t]$.
Thus one can obtain $\tf_{t,L_t[i]} = select_1(T_t,v_t-rank_1(R_t,i)+1)$.
The $s_t$ sequence is also represented using a bitmap $S[1,n]$ providing
$rank$/$select$ operations. Thus we can recover $s_t = select_1(S,t)$, and also
$rank_1(S,i)$ tells us which list $L[i]$ belongs to.

Let us analyze the space required by our representation.
According to Lemma~\ref{lem:wt}, the wavelet tree of $L$ occupies is
$n\log m + \Oh{n}$ bits.
The classical encoding of inverted files, when documents are sorted by
increasing document identifier, records the consecutive differences using
$\delta$-codes \cite{WMB99}. This needs at most
$\sum_t \df_t \log \frac{m}{\df_t} \le n\log\frac{m\nu}{n}$ bits plus
lower-order terms, which is asymptotically less than our space. If, however,
the lists are sorted by decreasing $\tf$ values, then differential encoding
can only be used on some parts of the lists. Yet, $n\log m$ (plus lower-order
terms) is still an upper bound to the space required to list the documents. As
can be seen, no inverted index representation takes more space than our wavelet
tree. However, it must be remembered that our wavelet tree will offer the combined
functionality of {\em both} inverted indexes, and more.

We also store the $\tf$ and the $s_t$ values. The former is encoded
with $T_t$ and $R_t$. We use Okanohara and Sadakane's representation
\cite{OS07} for $T_t$ and P{\v{a}}tra\c{s}cu's \cite{Pat08} for $R_t$ (see
Section~\ref{sec:wt}), to achieve total space $v_t\log\frac{m_t}{v_t}+\Oh{v_t}+
v_t\log\frac{\df_t}{v_t} + o(\df_t)$ bits and retain constant time access to
$\tf$ values. This space is similar to that needed to represent, in a
traditional $\tf$-sorted index, each new $\tf_{t,d}$ value and the number of
entries that share it. The $s_t$ values require $\nu\log\frac{n}{\nu} + o(n)$
bits using again P{\v{a}}tra\c{s}cu \cite{Pat08}, which gives constant-time
access to $s_t$ and requires less space than the usual pointers from the
vocabulary to the list of each term.
Overall our data structure takes at most $n \log (m\nu) + \Oh{n}$ bits.

We will now consider the classical and extended operations that can be carried
out with our data structure. In particular we will show how to give some
support for hierarchical document retrieval (as already seen for general
documents) and for {\em stemmed} searches without using any extra space.
One common way to support stemming is by coalescing terms having the
same root at index construction time. However, the index is then unable to
provide non-stemmed searching.
One can of course index the stemmed and non-stemmed
occurrence of each term, but this costs space. Our method can provide
both types of search without using any extra space provided all the
variants of the same stemmed word be contiguous in the vocabulary (this is in
many cases automatic as stemmed terms share the same root, or prefix).

\subsection{Full-Text Retrieval}

The full-text index, rather than $L_t$, requires a list $F_t$, where the same
documents are sorted by increasing document identifier.
Different kinds of access operations need to be carried out on $F_t$. We now show
how all these can be supported in $\Oh{\log m}$ time or less.

\subsubsection{Direct retrieval}

First, with our wavelet tree representation of $L$ we can compute any specific
value $F_t[k]$ in time $\Oh{\log m}$. This is equivalent to finding the $k$th
smallest value in $L[s_t,s_{t+1}-1]$, that is, query $\rqq(L,s_t,s_{t+1}-1,k)$
described in Section~\ref{sec:rqq}.

We can also extract any segment $F_t[k,k']$, in order, in time $\Oh{(k'-k+1)
\log\frac{m}{k'-k+1}}$, that is, faster per document as we extract more
documents. The algorithm is as for $\rqq$ on quantiles $k$ to $k'$
simultaneously, going just by one branch when both $k$ and $k'$ choose the
same branch, and splitting the interval into two separate searches when they do not.
We arrive at $k'-k+1$ leaves of the wavelet tree, thus the cost follows from
Lemma~\ref{lem:nodeset}.

Another useful operation is {\em fingered search}, that is, to find $F_t[k']$
after having visited $F_t[k]$, for some $k'>k$.
This is slightly more complex than for consecutive range next value
queries. We need to store $\log m$ values $m_\delta$, $e_\delta$ and
$v^\delta$, where $m_0 = \infty$ and $e_1 = 0$, and the others are computed as
follows when we obtain $F_t[k]$: at wavelet tree node $v$ of depth $\delta$
(the root being depth 1) we set $v^\delta \leftarrow v$ and, if we must go
to the left child, then we set $m_\delta \leftarrow e_\delta+n_l$ and
$e_{\delta+1} \leftarrow e_\delta$; else we set $m_\delta \leftarrow
m_{\delta-1}$ and $e_{\delta+1} \leftarrow e_\delta+n_l$.
Here $n_l$ is the value local to the node (recall \textbf{rqq} in
Algorithm~\ref{alg:rqqrnvint}). Therefore $e_\delta$ counts the values
skipped to the left, and $m_\delta$ is the maximum $k'$ value such that the
downward paths to compute $F_t[k]$ and $F_t[k']$ coincide up to depth $\delta$.
Now, to compute $F_t[k']$, we consider all the $\delta$ values, from largest to
smallest, until finding the first one such that $k' \le m_\delta$. From
there on we recompute the downward path, resetting $m_\delta$, $e_\delta$,
and $v^\delta$ accordingly.

If we carry out this operation $r$ times, across a range $[k,k']$, the cost is
$\Oh{\log m + r\log\frac{k'-k+1}{r}}$ by Lemma~\ref{lem:leafrangeset}.
Algorithm~\ref{alg:rqqext} depicts the new extended variants of \textbf{rqq}.

\begin{algorithm}[t]
\caption{Extended variants of range quantile algorithms:
$\mathbf{mrqq}(v_{root},i,j,k,k')$ outputs all the (distinct) values
$\rqq(S,i,j,k)$ to $\rqq(S,i,j,k')$, with their frequencies, on the wavelet
tree of sequence $S$, assuming $k'\le j-i+1$;
$\mathbf{frqq1}(v_{root},i,j,k)$ returns the same as
$\mathbf{rqq}(v_{root},i,j,k)$ but prepares the iterator for subsequent
fingered searches; those are carried out by calling
$\mathbf{frqq1}(v_{root},k)$, where it is assumed that the $k$ values increase
at each call; $\mathbf{frqq'}$ is the recursive procedure that reprocesses the
needed part of the path.}
\label{alg:rqqext}
\begin{tabular}{ccc}
\begin{minipage}{0.36\textwidth}
$\mathbf{mrqq}(v,i,j,k,k')$
\begin{algorithmic}
\IF{$v$ is a leaf}
	\OUTPUT $(label(v),j-i+1)$
\ELSE
	\STATE $i_l \leftarrow rank_0(B_v,i-1)+1$
	\STATE $j_l \leftarrow rank_0(B_v,j)$
	\STATE $i_r \leftarrow i-i_l$, $j_r \leftarrow j-j_r$
	\STATE $n_l \leftarrow j_l-i_l+1$
	\IF{$k \le n_l$}
		\STATE $\mathbf{mrqq}(v_l,i_l,j_l,k,\min(n_l,k'))$
	\ENDIF
	\IF{$k' > n_l$}
		\STATE $\mathbf{mrqq}(v_r,i_r,j_r,\max(k-n_l,1),k')$
	\ENDIF
\ENDIF
\STATE
\STATE
\STATE
\end{algorithmic}
\end{minipage}
&
\begin{minipage}{0.28\textwidth}
$\mathbf{frqq1}(v,i,j,k)$
\home{\up}
\begin{algorithmic}
\STATE $m_0 \leftarrow \infty$
\STATE $e_1 \leftarrow v$
\STATE $i^* \leftarrow i$
\STATE $j^* \leftarrow j$
\RET $\mathbf{frrq'}(v,i,j,k,1)$
\STATE
\end{algorithmic}

$\mathbf{frqq}(v,k)$
\home{\up}
\begin{algorithmic}
\STATE $\delta \leftarrow \textrm{height of }v$
\WHILE{$k > m_{\delta-1}$}
	\STATE $\delta \leftarrow \delta-1$
\ENDWHILE
\RET $\mathbf{frqq'}(v^\delta,i^*,j^*,k,\delta)$
\STATE
\STATE
\STATE
\STATE
\home{\STATE}
\end{algorithmic}
\end{minipage}
&
\begin{minipage}{0.36\textwidth}
$\mathbf{frqq'}(v,i,j,k,\delta)$
\begin{algorithmic}
\IF{$v$ is a leaf}
	\OUTPUT $(label(v),j-i+1)$
\ELSE
	\STATE $v^\delta \leftarrow v$
	\STATE $i_l \leftarrow rank_0(B_v,i-1)+1$
	\STATE $j_l \leftarrow rank_0(B_v,j)$
	\STATE $i_r \leftarrow i-i_l$, $j_r \leftarrow j-j_r$,
	       $n_l \leftarrow j_l-i_l+1$
	\IF{$k \le n_l$}
		\STATE $m_\delta \leftarrow e_\delta+n_l$
		\STATE $e_{\delta+1} \leftarrow e_\delta$
		\RET $\mathbf{frqq'}(v_l,i_l,j_l,k,\delta+1)$
	\ELSE
		\STATE $m_\delta \leftarrow m_{\delta-1}$
		\STATE $e_{\delta+1} \leftarrow e_\delta + n_l$
		\RET $\mathbf{frrq'}(v_r,i_r,j_r,k,\delta+1)$
	\ENDIF
\ENDIF
\end{algorithmic}
\end{minipage}
\end{tabular}
\end{algorithm}

\subsubsection{Intersection algorithms}

The most important operation in the various list intersection algorithms
described in the literature is to find the first $k$ such that $F_t[k] \ge d$, given $d$.
This is usually solved with a combination of sampling and linear, exponential,
or binary search. In our case, this operation takes time $\Oh{\log m}$ with query
$\rnv(L,s_t,s_{t+1}-1,d)$ described in Section~\ref{sec:rnv}. Our time
complexity is not far from the $\Oh{\log (s_{t+1}-s_t)}$ of traditional
approaches. Moreover, as explained in Section~\ref{sec:int}, we can use
fingered searches on $\mathbf{rnv}$ to achieve time
$\Oh{\log m + r \log \frac{m}{r}}$ for $r$ accesses. Furthermore, if all the
accesses are for documents in a range $[d,d']$ then, by
Lemma~\ref{lem:leafrangeset}, the cost will be
$\Oh{\log m + r \log \frac{d'-d+1}{r}}$ time. This is indeed the
time required by $r$ successive searches using exponential search.

Finally, we can intersect the lists $F_t$ and $F_{t'}$ using
$\rint(L,s_t,s_{t+1}-1,s_{t'},s_{t'+1}-1)$, in adaptive time
$\Oh{\alpha \log\frac{m}{\alpha}}$ --- recall Section~\ref{sec:int}.
As explained, this can be extended to intersecting $k$ terms simultaneously,
and to report documents where a minimum number of the terms appear.

\subsubsection{Other operations of interest}

If the range of terms $[t,t']$ represent the derivatives of a single stemmed
root, we might wish to act as if we had a single list $F_{t,t'}$ containing
all the documents where they occur.
Indeed, if we apply our previous algorithm to obtain $F_t[k]$ from
$L[s_t,s_{t+1}-1]$, on the range $L[s_t,s_{t'+1}-1]$, we obtain precisely
$F_{t,t'}[k]$, if we understand that a document $d$ may repeat several times
in the list if different terms in $[t,t']$ appear in $d$. Still we can obtain
the list of $docc$ distinct documents for a range of terms $[t,t']$ with
exactly the same method as for the $D$ array, described at the beginning
of Section~\ref{sec:doclist}, in time $\Oh{docc\log\frac{m}{docc}}$.

Furthermore, the algorithms to find the first $k$ such that $F_t[k] \ge d$,
can be applied verbatim to obtain the same result for $F_{t,t'}[k] \ge d$.
All the variants of these queries are directly supported as well.
Our intersection algorithm can also be applied verbatim in order to
intersect stemmed terms.

Additionally, note that we can compute some {\em summarization} information.
More precisely, we can obtain the {\em local vocabulary} of a document $d$,
that is, the set of different terms that appear in $d$. By executing
$rank_1(S,select_d(L,i))$ for successive $i$ values, we obtain all the local
vocabulary, in order, and in time $\Oh{\log m}$ per term.
This allows, for example, merging the vocabularies of
different documents, or binary searching for a particular term in a particular
document (yet, the latter is easier via two $rank$ operations on $L$:
$rank_d(L,s_{t+1}-1)-rank_d(L,s_t-1)$; then the corresponding position can be
obtained by $select_d(L,1+rank_d(L,s_t-1))$).

Finally, the data structure provides some basic support for temporal and hierarchical
documents, by restricting the inverted lists $F_t$ to a range of document
values $[d_{min},d_{max}]$ (recall Section~\ref{sec:temphier}).
A simple way to proceed is to first carry out a
query $\rnv(L,s_t,s_{t+1}-1,d_{min})$ with $\mathbf{rnv}$
(Algorithm~\ref{alg:rqqrnvint}), which will also give us the rank $p$ of the
first document $\ge d$. Then any subsequent range quantile query on $F_t$ must
increase its argument by $p-1$, and discard answers larger than $d_{max}$. On
the other hand, functions $\mathbf{hdlist}$ and $\mathbf{hdint}$
(Algorithm~\ref{alg:hier}) will work without changes, and support inverted
list algorithms on XML retrievable units, just as in Section~\ref{sec:units}.

\subsection{Ranked Retrieval}

We focus now on the operations of interest for ranked retrieval, which are
also simulated in $\Oh{\log m}$ time or less.

\subsubsection{Direct access and Persin's algorithm}

The $L_t$ lists used for ranked retrieval are directly concatenated in $L$, so
$L_t[i]$ is obtained by accessing symbol $L[s_t+i-1]$ using the wavelet tree.
Recall that the term frequencies $\tf$ are available in constant time.
A range $L_t[i,i']$ is obtained in time $\Oh{(i'-i+1)\log\frac{m}{i'-i+1}}$
by using query $\rr(L,s_t+i,s_t+i',[1,m])$ (Algorithm~\ref{alg:range}).

This algorithm has the problem of retrieving the documents
in document order, not in $\tf$ order as they are in $L_t$. Note, however, that
retrieving the highest-$\tf$ documents in document order is indeed beneficial
for Persin's algorithm \cite{PZSD96} (recall Section~\ref{sec:invind}), where a
problem is how to accumulate results across unordered document sets.
More precisely, assume we have the current candidate set as an array ordered
by increasing document identifier. Persin's algorithm computes a threshold
term frequency $f$, so that the next list to consider, $L_t$, should be
processed only for $\tf$ values that are at least $p$. Instead of traversing
$L_t$ by decreasing $\tf$ values and stopping when these fall below $f$, we
can compute $p = select_1(R_t,v_t-rank_1(T_t,f)+1)-1$, so that $L_t[1,p]$ is
precisely the prefix where the term frequencies are at least $f$. Now we
extract all the values as explained. As they are obtained in increasing
document identifier order, they are easily merged with the current candidate
set, in order to accumulate frequencies in common documents.


\subsubsection{Other operations of interest}

Any candidate document $d$ in Persin's algorithm can be directly evaluated,
obtaining its $\tf_{d,t}$ values, by finding $d$ within $L_t$ for each $t \in q$
(with $rank_d$ and $select_d$ on $L$, as explained), and its $\tf$ obtained
from $R_t$ and $T_t$, all in $\Oh{|q|\log m}$ time.

If we use stemming, we might want to retrieve prefixes of several
lists $L_t$ to $L_{t'}$. We may carry out the previous algorithm to deliver
all the distinct documents in these prefixes, now carrying on the $t'-t+1$
intervals as we descend in the wavelet tree. When we arrive at the relevant
leaves labeled $d$, the corresponding positions will be contiguous, thus we
can naturally return just one occurrence of each $d$ in the union.
If we wish to obtain the sum of the $\tf$ values for all the stemmed terms in
$d$, we can traverse the wavelet tree upwards for each interval element at
leaf $d$, and obtain its $\tf$ upon finding its position in $L$.
Alternatively, we could store the $\tf$ values aligned to the leaves and
mark their cumulative values on a compressed bitmap, so as to obtain the
sum in constant time as the difference of two $select_1$ operations on that
bitmap. The space for $\tf$, however, becomes now $n\log\frac{N}{n} + \Oh{n}$
bits, which is higher than in our current representation.
This method also delivers the results in document order.

Maintaining the $\tf$ values aligned to the leaf order yields some support for
hierarchical queries. Assume a retrievable unit (recall Section~\ref{sec:units})
spans the document range $[dl,dr]$, and thus we wish to compute the total $\tf$
of $t$ in range $[dl,df]$. Any such range is exactly covered by $\Oh{\log m}$
wavelet tree nodes (Lemma~\ref{lem:cover}). We can descend, projecting
the range of $L_t$ in $L$, until those nodes, and then add up the
accumulated $\tf$ values of those $\Oh{\log m}$ nodes, in overall time
$\Oh{\log m}$.

We can also support temporal and hierarchical documents by restricting our
accesses in $L_t$ only to documents within a range $[d_{min},d_{max}]$ (recall
Section~\ref{sec:temphier}). It is sufficient to use $[d_{min},d_{max}]$ as
the last argument when we use the $\rr$ query that underlies our support for
accessing $L_t$. This automatically yields, for example, Persin's
algorithm restricted to a range of documents.

\section{Conclusions}
\label{sec:concl}

The wavelet tree data structure \cite{GGV03} has had an enormous impact on the
implementation of reduced-space text databases. In this article we have shown
that it has several other under-explored capabilities. We have proposed
three new algorithms on wavelet trees that solve fundamental problems,
improving upon the state of the art in some aspects. For range intersections we
achieve an adaptive complexity that matches the one achieved for sorted
ranges. For range quantile and range next value problems, we match or get
close to the best known time complexities while using less space: basically
that needed to represent the sequence $S[1,n]$ plus $\Oh{n}$ extra bits, versus
the $\Oh{n\log n}$ extra bits required by previous solutions. Furthermore, if
we use compressed bitmap representations \cite{rrr2002} in our wavelet trees, we
retain the time complexities and achieve zero-order compression in the
representation of $S$ \cite{GGV03}, that is, our overall space including the
sequence becomes $nH_0(S) + \Oh{n+\sigma}$, where $[1,\sigma]$ is the alphabet
of $S$ and $H_0(S)$ is its empirical zero-order entropy.

We have also explored a number of applications of those novel algorithms to
two areas of Information Retrieval (IR): document retrieval on general string
databases, and inverted indexes. In both cases we obtained support for a
number of powerful operations without further increasing the space required
to support basic ones.

The algorithms are elegant and simple to implement, so they have the potential
to be useful in practice. Future work involves implementing them within an IR
framework and evaluating their practical performance. Although we have used
some theoretical data structures for handling bitmaps within convenient space
bounds, practical variants of $rank$/$select$-capable
plain and compressed bitmaps, as well as various wavelet tree implementations,
are publicly available\footnote{See for example {\tt http://www.recoded.cl}.}.
Some preliminary experiments \cite{CNPT10} show that an early version of our
results \cite{GPT09} do improve significantly in practice upon the previous state of the art
on document retrieval for general strings. Our improved versions presented in
this article should widen the gap. In the case of inverted indexes we do not
expect our representation to be faster for the basic operations, yet it is
likely that it requires less space than that of a full-text plus a
ranked-retrieval inverted index, and that it is more efficient on
sophisticated operations.



\paragraph*{Aknowledgements.}
We thank J\'er\'emy Barbay for his help in understanding the adaptive
complexity measures for intersections, and Meg Gagie for righting our
grammar.

\bibliographystyle{plain}
\bibliography{wavelet}

\end{document}